\documentclass[10pt,conference]{IEEEtran} 

\usepackage{
  amsmath,
  amsfonts,
  booktabs,
  pifont, 
  amsthm,
  algorithm,
  algpseudocode,
  caption,
  hyperref,
  cleveref,
  colortbl,
  paralist, 
  enumitem,
  float,
  graphicx,
  multicol,
  listings,
  subcaption,
  textcomp,
  tikz,
  xtab,
  xcolor,
  wrapfig
}

\usetikzlibrary{
  positioning, automata, arrows,
  arrows.meta, shapes.geometric, fit
}

\theoremstyle{definition}
\newtheorem{theorem}{Theorem}

\newtheorem{definition}{Definition}
\newtheorem{example}{Example}

\lstdefinelanguage{CodeQL}{
  basicstyle=\scriptsize\ttfamily,
  numbers=left,
  breaklines=true,
  comment=[l]{//},
  morecomment=[s]{/**}{*/},
  emph={    
    forall, exists, then, and, or, not, predicate, from, where,
    select, forex, count, class, extends
  },
  emphstyle={\textbf}
}

\newcommand{\threads}{\ensuremath{\mathcal{T}}}
\newcommand{\operations}{\ensuremath{\mathcal{O}}}
\newcommand{\actions}{\ensuremath{\mathcal{A}}}
\newcommand{\executions}{\ensuremath{\mathcal{E}}}
\newcommand{\fieldaccessactions}{\ensuremath{\actions_a}}
\newcommand{\syncactions}{\ensuremath{\actions_s}}
\newcommand{\otheractions}{\ensuremath{\actions_o}}

\newcommand{\javaclasses}{\ensuremath{\mathcal{C}}}
\newcommand{\classfields}{\ensuremath{\mathcal{F}}}

\newcommand{\classmethods}{\ensuremath{\mathcal{M}}}

\newcommand{\readfield}{\ensuremath{\mathit{read}}}
\newcommand{\writefield}{\ensuremath{\mathit{write}}}
\newcommand{\lock}{\ensuremath{\mathit{lock}()}}
\newcommand{\unlock}{\ensuremath{\mathit{unlock}()}}
\newcommand{\volatile}{\ensuremath{\mathit{volatile}}}
\newcommand{\private}{\ensuremath{\mathit{private}}}
\newcommand{\final}{\ensuremath{\mathit{final}}}

\newcommand{\programorder}{\ensuremath{\to_{\mathit{po}}}}

\newcommand{\syncorder}{\ensuremath{\to_{\mathit{syn}}}}
\newcommand{\initorder}{\ensuremath{\to_{\mathit{ini}}}}

\newcommand{\code}[1]{\texttt{\small #1}}

\newcommand{\proponeref}{\ref{prop:no_escaping}}
\newcommand{\proptworef}{\ref{prop:safe_publication}}
\newcommand{\propthreeref}{\ref{prop:correctly_synchronized}}

\begin{document}

\title{Scalable Thread-Safety Analysis of Java Classes with CodeQL}

\author{\IEEEauthorblockN{Bj{\o}rnar Haugstad J{\aa}tten}
\IEEEauthorblockA{
\textit{IT University of Copenhagen}\\
Copenhagen, Denmark \\
bjja@itu.dk}
\and
\IEEEauthorblockN{Simon Boye J{\o}rgensen}
\IEEEauthorblockA{
\textit{IT University of Copenhagen}\\
Copenhagen, Denmark \\
sboj@itu.dk}
\and
\IEEEauthorblockN{Rasmus Petersen}
\IEEEauthorblockA{
\textit{CodeQL/GitHub}\\
Copenhagen, Denmark \\
yoff@github.com}
\and
\IEEEauthorblockN{Ra{\'u}l Pardo}
\IEEEauthorblockA{
\textit{IT University of Copenhagen}\\
Copenhagen, Denmark \\
raup@itu.dk}
}

\maketitle


\begin{abstract}
In object-oriented languages software developers rely on \emph{thread-safe classes} to implement concurrent applications.
However, determining whether a class is thread-safe is a challenging task.
This paper presents a highly scalable method to analyze thread-safety in Java classes.
We provide a definition of thread-safety for Java classes founded on the correctness principle of the Java memory model, data race freedom.
We devise a set of properties for Java classes that are proven to ensure thread-safety.
We encode these properties in the static analysis tool CodeQL to automatically analyze Java source code.
We perform an evaluation on the top 1000 GitHub repositories.
The evaluation comprises 3632865 Java classes; with 1992 classes annotated as \texttt{@ThreadSafe} from 71 repositories.
These repositories include highly popular software such as Apache Flink (24.6k stars), Facebook Fresco (17.1k stars), PrestoDB (16.2k starts), and gRPC (11.6k starts).
Our queries detected thousands of thread-safety errors.
The running time of our queries is below 2 minutes for repositories up to 200k lines of code, 20k methods, 6000 fields, and 1200 classes.
We have submitted a selection of detected concurrency errors as PRs, and developers positively reacted to these PRs.
We have submitted our CodeQL queries to the main CodeQL repository, and they are currently in the process of becoming available as part of GitHub actions.
The results demonstrate the applicability and scalability of our method to analyze thread-safety in real-world code bases.

\end{abstract}

\section{Introduction}\label{sec:intro}

Many modern software projects require the use of concurrency, either for scalability or performance or both.
In object-oriented languages, such as Java, software developers rely on \emph{thread-safe classes} to implement concurrent applications.
A thread-safe class must guarantee that instances of the class behave correctly when used by multiple threads, without the need for any external synchronization.
The use of thread-safe classes is ubiquitous in object-oriented concurrent applications~\cite{jcp,herlihy2020art}, as it helps developers ensure the correctness of concurrent applications in a modular fashion.

However, determining whether a class is thread-safe is a challenging task.
Internet fora such as Stack Overflow and issue trackers in software repositories are full of questions regarding the thread-safety of classes and their use.

As a consequence, many methods have been studied to test or analyze whether a class is thread-safe.
One line of work focuses on testing~\cite{nistor2012ballerina,choudhary2017efficient,DBLP:conf/pldi/PradelG12,DBLP:conf/asplos/BurckhardtKMN10,DBLP:conf/asplos/ZhaoW0R25}.
Detecting concurrency errors via testing is difficult, as bugs appear non-deterministically depending on thread scheduling.
Standard unit testing cannot provide guarantees of finding thread-safety bugs, as it relies on the scheduler producing an execution that triggers the bug.
This has motivated work on testing using instrumented schedulers that attempt to explore all executions.
Since generating all possible executions of a concurrent program is infeasible, recent work on testing concurrency focuses on designing schedulers that probabilistically explore executions from the set of all executions.
In this way, it is possible to provide statistical bounds on the probability of finding thread-safety bugs.
Another line of work in the field of formal verification focuses on the use of model-checking and deductive verification to verify concurrency errors.
Model-checking attempts to analyze all possible executions of a program, and thus faces the same challenge as testing with non-deterministic schedulers.
As a consequence, recent works explore the use of abstractions that reduce the number of executions while maintaining the ability to find errors~\cite{DBLP:conf/cav/WuHHLS23, schemmel2020symbolic, aronis2018optimal}.
Even with these abstractions, efficient analysis of real software is not possible via these methods.
Deductive verification methods such as Vercors~\cite{DBLP:conf/esop/LeinoM09,DBLP:conf/fosad/LeinoMS09,DBLP:conf/plpv/AmighiBHZ12}, work directly on source code.
These methods transform the analysis of concurrency errors into an SMT problem~\cite{smt}.
Although this method allows for modular reasoning, it requires annotating source code with specifications that are rare in real world software and require in-depth knowledge of the underlying formalisms to utilize.
While SMT solving technology has been used on instances of real world software, improving its scalability is still a challenge~\cite{DBLP:conf/sigsoft/MikekZ23}.

This paper presents a novel and highly scalable method to analyze thread-safety in Java classes.
The Java memory model defines thread-safety for programs as lack of data races~\cite{jmm}, thread-safe programs are known as \emph{correctly synchronized}.
In a nutshell, correctly synchronized programs ensure that all variable accesses are ordered by the \emph{happens-before} relation~\cite{jmm}.
We lift this definition to the level of Java classes.
We provide a definition of correctly synchronized Java classes founded on the principles of the Java memory model.
Furthermore, we devise a set of rules that are sufficient to prove that a Java class is correctly synchronized.
We design an automatic thread-safety analysis method for Java classes by encoding these rules in the logic programming language CodeQL.
Our method can automatically identify classes which have been marked as thread-safe by developers (using the \code{@ThreadSafe} annotation from the Checker Framework), but are not thread-safe according to the Java memory model.
For each violation, the method reports the conflicting program statements that will cause a data race when executed concurrently.

The method is highly scalable, as part of our evaluation we analyze the top 1000 Java repositories in GitHub. 
We have used the CodeQL Multi-Repository Variant Analysis (MRVA) tool to run our evaluation.
The evaluation comprises 3632865 Java classes; with 1992 classes annotated as \texttt{@ThreadSafe} from 71 repositories.
These repositories include highly popular software such as Apache Flink (24.6k stars), Facebook Fresco (17.1k stars), PrestoDB (16.2k starts), and gRPC (11.6k starts) to mention a few.
Our queries detected thousands of thread-safety errors.
The running time of our queries is below 2 minutes for repositories up to 200k lines of code, 20k methods, 6000 fields, and 1200 classes; which demonstrates its applicability to real-world Java code repositories.
We have submitted a selection of detected concurrency errors as PRs, and developers have positively reacted to these PRs.
The theoretical analysis ensures that the method produces no false positives.
However, our evaluation detected a small fraction of false positives.
These false positives originated from the large diversity of thread-safe classes and (non-standard) locking mechanisms in the concurrent Java ecosystem.
That said, our queries can be extended to include new classes or packages containing thread-safe classes, and locking mechanisms.
We have submitted our CodeQL queries to the main CodeQL repository, and they are currently in the process of becoming available as part of GitHub actions. 

In summary, the contributions of this paper are:

\begin{itemize}
    \item A formal characterization of correctly synchronized classes founded on the principles of the Java memory model (Section \ref{sec:thread_safe_class}).
    \item An implementation in CodeQL of a method to automatically analyze thread-safety of Java classes. (Section \ref{sec:codeql})
    \item A large evaluation of our method on the top 1000 GitHub Java repositories, which demonstrates the scalability and applicability of our method in real world Java software (Section \ref{sec:evaluation}).
\end{itemize}

The repository with the CodeQL queries is publicly available at \url{https://github.com/itu-square/codeql-concurrency}.


\section{Java Memory Model}
\label{subsec:jmm_background}

We present the details of the Java memory model~\cite{jmm} that we use in \cref{sec:thread_safe_class} to define thread-safety for Java classes.

\subsubsection*{Notation}
Let $\javaclasses$ denote the universe of Java classes (or, simply, classes). A class $C \in \javaclasses$ is composed of a set of fields $\classfields_C$ (also known as attributes) and a set of methods $\classmethods_C$.
Let $\threads$ denote a set of threads and $\operations$ a set of operations.
Given $t \in \threads$ and $o \in \operations$, an \emph{action} $(t,o)$ models the execution of an operation $o$ by a thread $t$.
The set of actions is $\actions \subseteq \threads \times \operations$.
In order to define how actions in the Java memory model synchronize they are split into three types:

\begin{itemize}
\item 
  Field access actions, $(t,o) \in \fieldaccessactions$, correspond to actions involving read and write operations on a field. 
  Given a class field $f \in \classfields_C$, we use the notation $\readfield(f)$ and $\writefield(f)$ to denote read and write field access operations.
\item Synchronization actions, $(t,o) \in \syncactions$, which include \lock / \unlock\ operations on a monitor, and read and write accesses on \volatile\ fields.
\item Other actions, $(t,o) \in \otheractions$, i.e., any other operation not included in the above sets. For instance, actions involving fields that are not in $\classfields_C$ (e.g., actions on method local fields) or input/output operations.
\end{itemize}
As expected, $\actions = \fieldaccessactions \cup \syncactions \cup \otheractions$.
An execution $e \in \executions$ is a (possibly infinite) sequence of actions $e = a_1, a_2, \ldots$ such that $a_i$ is executed before $a_{i+1}$.

Given a (concurrent) program, the Java memory model defines the set of \emph{valid} executions of the program; i.e., the set of executions that the JVM is allowed to produce for the program.
These executions are characterized by two relations: \emph{happens-before} and \emph{synchronization order}.
These relations impose constraints on the order in which the actions in executions must take place and their effect on visibility (i.e., what value a thread must see in a read access).
In the following we present the details of these relations.

\subsubsection*{Happens-before}
The happens-before relation is a transitive closure of the union of two relations: \emph{program order} and \emph{synchronizes-with}.
The program order relation captures \emph{intra-thread} semantics (i.e., the execution semantics when the program is executed in isolation by a single thread), and 
the synchronizes-with relation captures the order in which synchronization actions synchronize between threads.

\begin{definition}[Happens-before]
For any two actions $a, b \in \actions$ in an execution of a concurrent program, we say that $a$ happens-before $b$, denoted as $a \to b$ if and only if
\begin{enumerate}[label=HB\arabic*]

\item 
  \textit{Program order:} Action $a$ is executed before be $b$ according to program order (i.e., as executed by a single thread in isolation). \label{hb_rule:program_order}

\item 
  \textit{Monitor lock:} Action $a$ is an \unlock\ and $b$ is any subsequent \lock\ on the same monitor. \label{hb_rule:monitor}

\item 
  \textit{Volatile order:} Given a \volatile\ field $f$, action $a = \writefield(f)$ is a write on a volatile field and $b = \readfield(f)$ is a subsequent read. \label{hb_rule:volatile}

\item
  \textit{Default value init:} Action $a$ writes the default value of a field during object initialization and $b$ is the first action in any other thread. \label{hb_rule:default_init}

\item
  \textit{Final field init:} Action $a$ writes the initial value of \texttt{final} field during object initialization and $b$ is the first action in any other thread. \label{hb_rule:final_init}

\item 
  \textit{Transitivity:} Given $c \in \actions$, if $a \to c$ and $c \to b$, then $a \to b$. \label{hb_rule:transitivity}
\end{enumerate}

We use $\programorder$, $\syncorder$ and $\initorder$ to distinguish the relations related to program order (\ref{hb_rule:program_order}), synchronization (\ref{hb_rule:monitor}, \ref{hb_rule:volatile}) and initialization (\ref{hb_rule:default_init}, \ref{hb_rule:final_init}). As expected, $\programorder \cup \syncorder \cup \initorder \; \subseteq \; \to$.
\end{definition}

\paragraph*{Synchronization order}
The synchronization order is a total order between the synchronization actions in an execution.
A synchronization order satisfies the following properties:
\begin{enumerate}[label=\arabic*.]
\item Synchronization order is consistent with mutual exclusion. This implies that lock and unlock synchronization actions are correctly nested.
\item Synchronization order is consistent with the program order.
\end{enumerate}
Programs may give rise to different synchronization orders; one for each interleaving of synchronization actions.

Only executions whose sequence of actions satisfy program order, synchronization order and happens-before are valid, or, \emph{well-defined}.
Formally,
\begin{definition}[Well-formed executions]
An execution $e \in \executions$ is \emph{well-formed} iff it is consistent with program order, synchronization order and happens-before.
\label{def:well-formed-execution}
\end{definition}

The notion of correctness for concurrent programs that the Java memory model adopts is data-race freedom~\cite{jmm}.
If a program is data-race free, then it is also categorized as \emph{correctly synchronized}.
In what follows, we provide the formal details of these concepts.
We start by defining \emph{conflicting actions}. 
Intuitively, two actions are conflicting if they read and write on the same class field.

\begin{definition}[Conflicting actions]
We say that actions $a, b \in \actions$ in an execution are \emph{conflicting} iff they are field access operations on the same class field $f$, and at least one of them is a write, i.e., $a = \readfield(f)$ and $b=\writefield(f)$.
\label{def:conflicting_actions}
\end{definition}
We remark that this definition refers to reads and writes on non-volatile fields.
Read and write access to volatile fields are never conflicting according to the Java memory model~\cite{jmm}.

If conflicting actions are executed concurrently, it \emph{may} lead to concurrency issues.
The lack of synchronization in conflicting actions leads to \emph{data races}.

\begin{definition}[Data race]
A \emph{data race} exists between two actions $a, b \in \actions$ in an execution \(e \in \executions\) iff they are conflicting actions and they are not ordered by happens-before.
\label{def:data_race}
\end{definition}

\begin{example}
  Here we provide a simple example of a program with a data race.
  Consider the class \texttt{CounterDR} in the top-left of~\cref{fig:method_summary_and_examples}.
  Now consider two threads that concurrently execute the method \texttt{c.inc();} for the same object \texttt{c} of type \texttt{CounterDR}.
  We can use the Java memory model to show that there exist a data race in this program.
  
  Let $(t_i,a)$ denote action $a$ executed by thread $t_i$.
  We use $(n)$ to refer to the program statement in line $n$.
  First, observe that there are two conflicting actions, namely, $(t_i, (5))$ and $(t_j, (7))$ for $i \not= j$.
  To show that there is a data race, we must find (at least) one execution of this program where the pair $(t_i, (5))$ and $(t_j, (7))$ is not ordered by happens-before.
  From program order, we have, for all executions, the following pairs of actions in the happens-before relation: $(t_i,(n)) \programorder (t_i,(n+1))$ for $n \in \{5,6\}$.
  Since there are no synchronization actions in this program, the set of synchronization orders is empty and there are no more pairs of actions in the happens-before relation.
  As a consequence, it is not possible to derive $(t_i, (5)) \to (t_j, (7))$ where $i \not= j$ directly or from the transitive closure of happens-before.
  Furthermore, the above implies that executions are solely characterized by interleaving of actions that satisfies program order.
  For instance, the execution $t_1(5),t_2(5),t_1(6),t_2(6),t_1(7),t_2(7)$ is well-formed and neither $t_1(5) \to t_2(7)$ nor $t_2(5) \to t_1(7)$ holds.
  Thus, we found an execution with conflicting actions that are not part of the happens-before order, which, by~\cref{def:data_race}, implies that there is a data race in the execution.
  \label{ex:data_race}
\end{example}

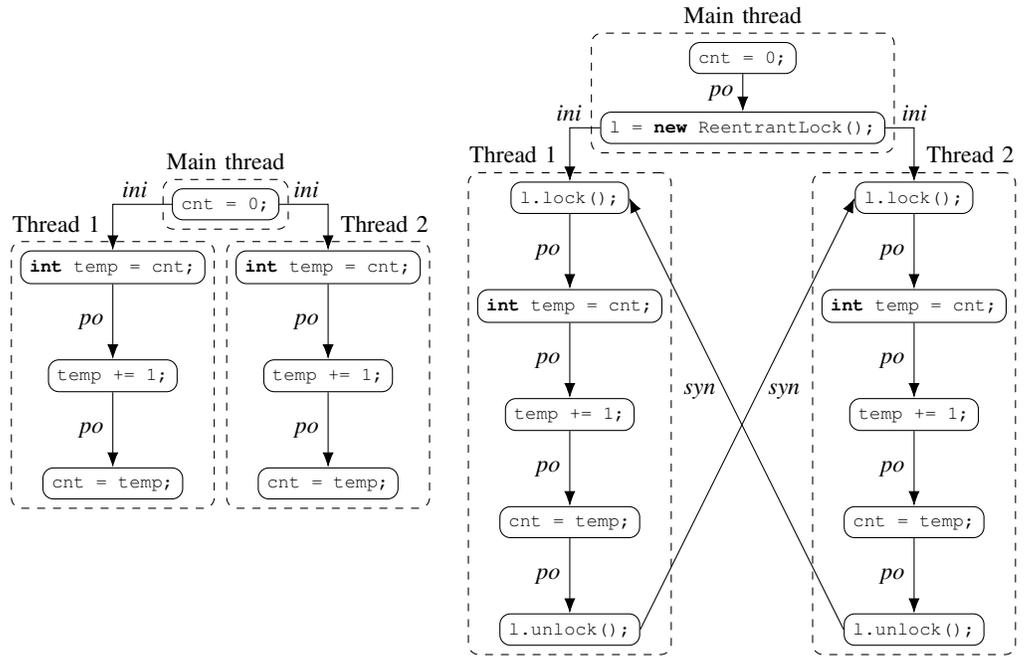
\begin{figure*}[t]
  \centering
  \begin{minipage}{.24\textwidth}
  \begin{lstlisting}[
  language=java,
  basicstyle=\scriptsize\ttfamily
  ]
  class CounterDR {
    int cnt = 0;
  
    public void inc() {
      int temp = cnt;
      temp += 1;
      cnt = temp;
    }
  }
  \end{lstlisting}
  \vspace{3mm}
  \begin{lstlisting}[
  language=java,
  basicstyle=\scriptsize\ttfamily
  ]
  class CounterTS {
    private int cnt = 0;
    private final Lock l = 
      new ReentrantLock();
  
    public void inc() {
      l.lock();
      int temp = cnt;
      temp += 1;
      cnt = temp;
      l.unlock();
    }
  }
  \end{lstlisting}
  \end{minipage}
  \hfill
  \begin{minipage}{.32\textwidth}
  \begin{tikzpicture}[
    every node/.style={draw, rounded corners, align=center},
    arrow/.style={-{Latex[length=2mm]}, thin},
    labelnode/.style={draw=none, font=\small}
  ]
  
  \node (1_stmt1) {\lstinline|int temp = cnt;|};
  \node (1_stmt2) [below=of 1_stmt1] {\lstinline|temp += 1;|};
  \node (1_stmt3) [below=of 1_stmt2] {\lstinline|cnt = temp;|};
  
  \draw [arrow] (1_stmt1) -- node[labelnode, left] {\textit{po}} (1_stmt2);
  \draw [arrow] (1_stmt2) -- node[labelnode, left] {\textit{po}} (1_stmt3);
  
  \node[draw, dashed, fit=(1_stmt1) (1_stmt2) (1_stmt3), label={[xshift=-7.5mm]above:\small{Thread 1}}] {};
  
  \node (2_stmt1) [right=4mm of 1_stmt1] {\lstinline|int temp = cnt;|};
  \node (2_stmt2) [below=of 2_stmt1] {\lstinline|temp += 1;|};
  \node (2_stmt3) [below=of 2_stmt2] {\lstinline|cnt = temp;|};
  
  \draw [arrow] (2_stmt1) -- node[labelnode, left] {\textit{po}} (2_stmt2);
  \draw [arrow] (2_stmt2) -- node[labelnode, left] {\textit{po}} (2_stmt3);
  
  \node[draw, dashed, fit=(2_stmt1) (2_stmt2) (2_stmt3), label={[xshift=7.5mm]above:\small{Thread 2}}] {};
  
  \node (3_stmt1) [above=4mm of 1_stmt1, xshift=15mm] {\lstinline|cnt = 0;|};
  
  \draw [arrow] (3_stmt1) -| node[labelnode,yshift=2mm,right] {\textit{ini}} (1_stmt1);
  \draw [arrow] (3_stmt1) -| node[labelnode,yshift=2mm,left] {\textit{ini}} (2_stmt1);
  
  \node[draw, dashed, fit=(3_stmt1), label={[xshift=0mm]above:\small{Main thread}}] {};

  \end{tikzpicture}
  \end{minipage}
  \hfill
  \begin{minipage}{.41\textwidth}
  \begin{tikzpicture}[
    every node/.style={draw, rounded corners, align=center},
    arrow/.style={-{Latex[length=2mm]}, thin},
    labelnode/.style={draw=none, font=\small}
  ]
  
  \node (1_lock) {\lstinline|l.lock();|};
  \node (1_stmt1) [below=of 1_lock]{\lstinline|int temp = cnt;|};
  \node (1_stmt2) [below=of 1_stmt1] {\lstinline|temp += 1;|};
  \node (1_stmt3) [below=of 1_stmt2] {\lstinline|cnt = temp;|};
  \node (1_unlock) [below=of 1_stmt3] {\lstinline|l.unlock();|};
  
  \draw [arrow] (1_lock) -- node[labelnode, left] {\textit{po}} (1_stmt1);
  \draw [arrow] (1_stmt1) -- node[labelnode, left] {\textit{po}} (1_stmt2);
  \draw [arrow] (1_stmt2) -- node[labelnode, left] {\textit{po}} (1_stmt3);
  \draw [arrow] (1_stmt3) -- node[labelnode, left] {\textit{po}} (1_unlock);
  
  \node[draw, dashed, fit=(1_stmt1) (1_stmt2) (1_stmt3) (1_lock) (1_unlock), label={[xshift=-7.5mm]above:\small{Thread 1}}] {};
  
  \node (2_lock) [right=3cm of 1_lock]{\lstinline|l.lock();|};
  \node (2_stmt1) [below=of 2_lock] {\lstinline|int temp = cnt;|};
  \node (2_stmt2) [below=of 2_stmt1] {\lstinline|temp += 1;|};
  \node (2_stmt3) [below=of 2_stmt2] {\lstinline|cnt = temp;|};
  \node (2_unlock) [below=of 2_stmt3] {\lstinline|l.unlock();|};
  
  \draw [arrow] (2_lock) -- node[labelnode, left] {\textit{po}} (2_stmt1);
  \draw [arrow] (2_stmt1) -- node[labelnode, left] {\textit{po}} (2_stmt2);
  \draw [arrow] (2_stmt2) -- node[labelnode, left] {\textit{po}} (2_stmt3);
  \draw [arrow] (2_stmt3) -- node[labelnode, left] {\textit{po}} (2_unlock);
  
  \node[draw, dashed, fit=(2_stmt1) (2_stmt2) (2_stmt3) (2_lock) (2_unlock), label={[xshift=7.5mm]above: \small{Thread 2}}] {};
  
  \node (3_stmt2) [above=5mm of 1_lock, xshift=23mm] {\lstinline|l = new ReentrantLock();|};
  \node (3_stmt1) [above=5mm of 3_stmt2] {\lstinline|cnt = 0;|};
  
  \draw [arrow] (3_stmt1) -- node[labelnode, left] {\textit{po}} (3_stmt2);
  \draw [arrow] (3_stmt2) -| node[labelnode,yshift=2mm,xshift=-3mm,right] {\textit{ini}} (1_lock);
  \draw [arrow] (3_stmt2) -| node[labelnode,yshift=2mm,xshift=3mm,left] {\textit{ini}} (2_lock);
  
  \node[draw, dashed, fit=(3_stmt1) (3_stmt2), label={above:\small{Main thread}}] {};

  \draw [arrow] (1_unlock.east) -- node[labelnode,yshift=3mm, xshift=-3mm,left] {\textit{syn}} (2_lock.west);
  \draw [arrow] (2_unlock.west) -- node[labelnode,yshift=3mm, xshift=3mm,right] {\textit{syn}} (1_lock.east);
  \end{tikzpicture}
  \end{minipage}

  \caption{
    \underline{Left-top}: Non thread-safe implementation of a counter class \textit{CounterDR}.
    \underline{Left-bottom}: Thread-safe implementation of a counter class \textit{CounterTS}.
    \underline{Middle}: Pairs of actions related by happens-before in a concurrent execution of \textit{inc()} for \textit{CounterDR}. 
    \underline{Right}: Pairs of actions related by happens-before in a concurrent execution of \textit{inc()} for \textit{CounterTS}. 
    For the concurrent executions, a \emph{main thread} creates an instance of a counter object, creates two threads and each of those executes the method \textit{inc()}.
    The dashed boxes indicate the operations executed by each thread.
  }
  \label{fig:method_summary_and_examples}
\end{figure*}


\section{Thread-Safety in Java Classes}
\label{sec:thread_safe_class}

Goetz et al.~\cite{jcp}---one of the most popular reference textbooks for concurrent programming in Java---describes a thread-safe class as follows:
\begin{definition}[Thread-safe Java Class~\cite{jcp}]
\textit{``A class is thread-safe if it behaves correctly when accessed from multiple threads, regardless of the scheduling or interleaving of the execution of those threads by the runtime environment, and with no additional synchronization or other coordination on the part of the calling code.''}
\label{def:thread-safe-class-goetz}
\end{definition}

Although this definition provides an intuitive meaning for thread-safety, it is informal, as it does not precisely define what ``correct behaviour'' or ``access'' means.
The Java memory model defines correctly synchronized programs as those free from data races (cf.~\cref{subsec:jmm_background}).
Therefore, in this paper, we use this correctness notion for Java classes.
We lift the notion of correctly synchronized programs to classes as follows:
\begin{definition}[Correctly Synchronized Java Class]
A Java Class is \emph{correctly synchronized} iff no concurrent execution of field accesses and method calls results in data races.
\label{def:correctly-synchronized-class}
\end{definition}
It follows from~\cref{def:correctly-synchronized-class} that programs composed only of field accesses and methods calls on correctly synchronized classes are free from data races.
Since data race freedom is the notion of correctness that the Java memory model embraces, we equate correct synchronization to thread-safety for Java classes.
For a given class, the intended application logic may impose more constraints than data race freedom for it to be thread-safe, but here we only consider the general considerations enforced by the Java memory model.
In what follows, we use the terms \emph{correctly synchronized class} and \emph{thread-safe class} interchangeably.

The problem of determining whether a class is correctly synchronized can be formulated as follows:
Given a class $C$ with $\{m_1, \ldots, m_n\} = \classmethods_C$ and $\{f_1, \ldots, f_m\} = \classfields_C$, $C$ is correctly synchronized iff none of the possible executions of the program in Listing~\ref{lst:correctly_synchronized_check} result in a data race.

\begin{lstlisting}[
language=java,
basicstyle=\scriptsize\ttfamily,
escapeinside={@}{@},
label=lst:correctly_synchronized_check,
caption={For class \texttt{C} to be thread-safe, none of the well-formed executions of this program must contain data races.}
]
C c = new C(...);

// Field reads
new Thread(() -> {x@$_1$@ = c.f@$_1$@}).start();
...
new Thread(() -> {x@$_m$@ = c.f@$_m$@}).start();

// Field writes
new Thread(() -> {c.f@$_1$@ = v}).start();
...
new Thread(() -> {c.f@$_m$@ = v}).start();

// Method calls
new Thread(() -> {c.m@$_1$@()}).start();
...
new Thread(() -> {c.m@$_n$@()}).start();
\end{lstlisting}

Ensuring that a class is thread-safe is a difficult task for developers.
Goetz et al.\ propose a set of design guidelines that aim to ensure that a class is thread-safe~\cite{jcp}:
\begin{enumerate}[label=\textbf{G\arabic*}]
\item Class state must not escape. Class fields must only be accessed through the methods of the class. This guideline delimits the source code that the developer of a thread-safe class must consider to the methods of the class.
\item All class fields must be safely published. Safe publication refers to ensuring that class fields are correctly initialized before they are used by any other read (i.e., any thread reading the field must read the initial value if there was no write between the initialization and the read).
\item Synchronization must be used to access mutable fields.
\end{enumerate}

In what follows, we turn these guidelines into properties of a Java class.
Furthermore, we prove that that these properties are sufficient to ensure that a Java class is thread-safe according to~\cref{def:correctly-synchronized-class}.
\begin{enumerate}[label=\textbf{P\arabic*}]
\item \textit{No escaping}. All class fields $f \in \classfields$ must be declared as \private.\label{prop:no_escaping}
\item \textit{Safe publication}. All class fields $f \in \classfields$ must be either: i) initialized to the default value; ii) declared as \final; or iii) declared as \volatile.\label{prop:safe_publication}
\item \textit{Correctly synchronized}. For all field access actions $a, b \in \fieldaccessactions$, if they are conflicting on a field $f \in \classfields$ (cf~\cref{def:conflicting_actions}), then we have $l_1, l_2, u_1, u_2 \in \syncactions$ such that $l_i,u_i$ are \lock/\unlock\ operations on the same monitor, and it holds $l_1 \to a \to u_1$ and $l_2 \to b \to u_2$.\label{prop:correctly_synchronized}
\end{enumerate}

The following theorem proves Java classes with the properties above are correctly synchronized.

\begin{theorem}
If properties \ref{prop:no_escaping}, \ref{prop:safe_publication} and \ref{prop:correctly_synchronized} hold for a class $C$, then $C$ is correctly synchronized.
\end{theorem}

\begin{proof}

Let \(C \in \javaclasses\) be a correctly synchronized class.
Let \(e \in \executions\) be a well-formed execution composed by field accesses and methods calls on an object of class $C$.
Let \(f \in \classfields_C\) be a field in \(C\).
For any conflicting field accesses \(a, w\) in an execution \(e\) such that \(w = \writefield(f)\) and \(a \in \{\readfield(f), \writefield(f)\}\) we show that \(a, w\) are ordered by happens-before.

Conflicting actions can occur in the following cases:
\begin{inparaenum}[i)]
\item two direct field accesses executed by two different threads; \label{ca:1}
\item a direct field access and a method call involving \(a\) or \(w\) executed by two different threads; \label{ca:2}
\item an initialization write \(w_{ini}\) and a direct field access executed by two different threads; \label{ca:3}
\item an initialization write \(w_{ini}\) and a method call involving \(a\) executed by two different threads; \label{ca:4}
\item two method calls involving \(a\) and \(b\), respectively, and executed by two different threads; \label{ca:5}
\end{inparaenum}
We show that for each of these cases the conflicting accesses are ordered by happens-before or they cannot be part of an execution \(e\) composed only by operations on the correctly synchronized class \(C\).

By property~\ref{prop:no_escaping}, we have that threads cannot directly access fields in class $C$.
Therefore, we can conclude that cases \ref{ca:1}, \ref{ca:2} and \ref{ca:3} cannot occur in $e$.

For case~\ref{ca:4}, consider two threads \(t_{ini}, t\), denoting the thread initializing an object of class \(C\) and a thread executing a field access \(a\) via a method call, respectively.
Since $C$ is correctly synchronized and by property \ref{prop:safe_publication}, we have that \(w_{ini}\) accesses a field that is initialized to the default value, declared as \texttt{final} or declared as \texttt{volatile}.
If the field is declared as \texttt{volatile}, then the actions are not conflicting; recall that the definition of conflicting actions in the Java memory model excludes accesses on \texttt{volatile} variables.
If the field is initialized to the default value or declared as \texttt{final}, then, by happen-before rules \ref{hb_rule:default_init} and \ref{hb_rule:final_init} we have that \((t_{ini}, w_{ini}) \to (t, a_1)\) where \(a_1\) denotes the first operation executed by thread \(t\).
Since the method call is executed by \(t\), by rule \ref{hb_rule:program_order}, we have that \((t, a_1) \to (t, a)\).
Finally, by transitivity, we can derive \((t_{ini}, w_{ini}) \to (t,a)\), which implies that the conflicting actions in case~\ref{ca:4} are ordered by happens-before.

For case~\ref{ca:5}, consider two threads \(t_1, t_2\) executing two method calls each involving \(w\) and \(a\), respectively.
By property~\ref{prop:correctly_synchronized}, we have \(l_1 \to w \to u_1\) and \(l_2 \to a \to u_2\) where \(l_i, u_i\) are \textit{lock()}/\textit{unlock()} operations on the same monitor. 
Depending on the synchronization order, one of $l_1$ and $l_2$ is executed the earlier in $e$. 
Assume, without loss of generality, it is $l_1$. 
Then, by rule~\ref{hb_rule:monitor}, we have that \((t_1, u_1) \to (t_2, l_2)\). 
By transitivity (\ref{hb_rule:transitivity}) we can derive \((t_1, w) \to (t_2, a)\).
The case where $l_2$ occurs before $l_1$ is symmetric.
Thus, we can conclude that either \((t_1, w) \to (t_2, a)\) or \((t_2, a) \to (t_1, w)\); depending on the synchronization order.
This implies that the conflicting actions in case~\ref{ca:5} are ordered by happens-before.
\end{proof}

\begin{example}

Consider the two classes on the left of~\cref{fig:method_summary_and_examples}: \textit{CounterDR} and \textit{CounterTS}.
We illustrate the use of properties \ref{prop:no_escaping}-\ref{prop:correctly_synchronized} to show that the latter is thread-safe and the former is not.
The middle and right graphs in~\cref{fig:method_summary_and_examples} show the pairs of operations ordered by happens-before in a concurrent execution of the \textit{inc()} method---for \textit{CounterDR} and \textit{CounterTS}, respectively.
The concurrent execution is an instance of Listing~\ref{lst:correctly_synchronized_check} for the threads executing method calls.

We start with the class \textit{CounterDR}.
Since variable \textit{cnt} is initialized to 0 (the default value for integers in Java), we have that \ref{prop:safe_publication} holds.
In the middle graph in \cref{fig:method_summary_and_examples} we observe that the first operation of every thread and the initialization writes are ordered by $\initorder$; this is due to rule~\ref{hb_rule:default_init} of the happens-before relation.
There are conflicting actions in the execution, namely, the read access in line 5 and the write access in line 7.
Since there are no locks in the implementation of \textit{CounterDR}, \ref{prop:correctly_synchronized} does not hold.
The graph shows the operations ordered by $\programorder$ in each thread.
As discussed in~\cref{ex:data_race}, due to the absence of locks, there are not operations ordered by $\syncorder$.
As a consequence, we cannot order the conflicting accesses using happens-before.
This implies that the execution contains data races (cf.~\cref{def:data_race}), and consequently \textit{CounterDR} is not thread-safe (cf.~\cref{def:correctly-synchronized-class}).
Lastly, \ref{prop:no_escaping} does not hold.
The field \textit{cnt} is not defined as \textit{private}.
Recall that our definition of correctly synchronized class (cf.~\cref{def:correctly-synchronized-class}) requires absence of data races for concurrent field accesses.
However, in \textit{CounterDR}, \textit{cnt} can be read/written without synchronization by any thread.
Consider again Listing~\ref{lst:correctly_synchronized_check}, the field accesses part of the program lead to data races.

Consider now \textit{CounterTS} (left-bottom class and right graph in~\cref{fig:method_summary_and_examples}).
Class fields are declared as \textit{private}, thus \ref{prop:no_escaping} holds.
This prevents any data races from the field accesses part of Listing~\ref{lst:correctly_synchronized_check}; as those field accesses cannot be executed.
For \textit{CounterTS}, initialization assigns the default value 0 to \textit{cnt} and also initializes the lock object \textit{l}.
These two operations are executed by the main thread, and, consequently, they are ordered by $\programorder$.
Although \textit{l} is not initialized to its default value---the default value for non-primitive objects in Java is \texttt{null}---it is declared as \textit{final}.
Therefore, by applying rule \ref{hb_rule:final_init}, the initialization of \textit{l} and the first operation of every thread are ordered by $\initorder$; which implies absence of data races during initialization.
Having \textit{cnt} initialized to the default value and \textit{l} declared as \textit{final} also implies that \ref{prop:safe_publication} holds.
\textit{CounterTS} has the same set of conflicting operations than \textit{CounterDR}.
However, note that now all conflicting operations are surrounded by \textit{lock()}/\textit{unlock()} operations on the same lock (which is an explicit monitor).
This implies that \ref{prop:correctly_synchronized} holds.
Due to the happens-before rule \ref{hb_rule:monitor}, we have that \textit{lock()}/\textit{unlock()} operations are ordered by $\syncorder$.
(Any given execution will only have one of the two edges marked by $syn$ in the figure.)
Note that, in this graph, it is possible to find a path from/to any pair of conflicting operations.
This implies that using transitivity (\ref{hb_rule:transitivity}) we can establish a happens-before order between any two conflicting actions, and, consequently, it shows the absence of data races.
Not being able to directly access class fields and the lack of data races for the concurrent execution of method calls implies that the class is thread-safe.\qed

\end{example}


\section{Thread-Safety Analysis with CodeQL} \label{sec:codeql}
CodeQL is a logic programming language, which makes it relatively straightforward to express our properties P1-P3. Using the infrastructure provided by GitHub, CodeQL programs can easily be run on a vast amount of repositories. In this section, we explain how we write CodeQL predicates to express the concepts needed to capture incorrectly synchronized Java classes.

We will develop a series of predicates, culminating in one that contains exactly violations of P3. We will make heavy use of existing CodeQL libraries. These are designed to detect security vulnerabilities and contain many useful building blocks for static analysis. For example, at the level of syntax, we find predicates for recognizing field accesses, and at the level of control flow there is a binary predicate capturing when one control flow node dominates another.

CodeQL is object-oriented, which means that one can define CodeQL classes. A CodeQL class is defined by
writing its characteristic predicate and is thus simply a set of values. However, CodeQL classes offer some convenient syntax, for instance in the form of member predicates, which we illustrate by the following normal predicate, intended to capture field accesses that will need explicit synchronization:
\begin{lstlisting}[
    caption=Exposed field access,
    label=lst:exposed-field-access,
    language=CodeQL]
predicate exposed(FieldAccess a) {
  a.getField() = annotatedAsThreadSafe().getAField() and
  not a.getField().isVolatile() and
  not a.(VarWrite).getASource() = a.getField().getInitializer() and
  not a.getEnclosingCallable() = a.getField().getDeclaringType().getAConstructor()
}
\end{lstlisting}
In line 1, we declare that we are defining a new predicate, it is called "exposed", it is a unary predicate, and its values are from the outset restricted to values of the CodeQL class \code{FieldAccess}.
The body of the predicate, lines 2-5, is just a conjunction and, in keeping with the CodeQL style guide, we have written one conjunct on each line.

The class \code{FieldAccess} offers the member predicate \code{getField}, which is used on line 2. This is a binary predicate that relates values of \code{FieldAccess} to values of another CodeQL class, \code{Field}. The expression \code{a.getField()} will hold for all values of \code{Field} that are related to the particular value of \code{FieldAccess} denoted by \code{a}. 
The predicate \code{annotatedAsThreadSafe} holds for all Java classes, captured by the CodeQL class \code{Class}, that have the Java annotation \code{@ThreadSafe}. The expression \code{annotatedAsThreadSafe().getAField()} holds for all the values of \code{Field} that denote fields of Java classes with this annotation. The equality \code{a.getField() = annotatedAsThreadSafe().getAField()} holds if any value can be found on both the right hand side and the left hand side. Thus, line 2 expresses that \code{a} denotes an access to a field in a class with the \code{@Threadsafe} annotation.
CodeQL is evaluated bottom-up, meaning that all materialized predicates are materialized in full. It is therefor wise to restrict all predicates as much as possible, and we are only interested in field accesses to fields of annotated classes.

The class \code{Field} offers the member predicate \code{isVolatile} used on line 3. Thus, the conjunct on line 3 expresses that the field access denoted by \code{a} is not to a volatile field.

On line 4, we see a CodeQL type cast: \code{a.(VarWrite)} filters the possible values of \code{a} to those that are also values of the CodeQL class \code{VarWrite}. Line 4 expresses that if \code{a} is a write then it is not the one that writes the initial value.

Line 5 expresses that \code{a} is not found in the constructor for the Java class declaring the field being accessed.

Having identified the field accesses that might need synchronization, we wish to restrict our attention to those for the rest of the development. We can make this more convenient by wrapping those values in a CodeQL class:
\begin{lstlisting}[
    caption=CodeQL class ExposedFieldAccess,
    label=lst:class-exposed-field-access,
    language=CodeQL]
class ExposedFieldAccess extends FieldAccess {
  ExposedFieldAccess() { exposed(this) }
}
\end{lstlisting}
Line 1 declares that we are defining a new class, it is called "ExposedFieldAccess", and its values are a subset of the values of the CodeQL class \code{FieldAccess}. We also wish to inherit the member predicates from there.
Line 2 is the characteristic predicate of the class. It uses the special variable \code{this} to restrict the values that are in this class, specifically saying that they must all be exposed according to the predicate we defined earlier\footnote{In the full program, we have inlined \code{exposed} into the characteristic predicate, and it includes a few more cases such as fields where the type is thread-safe.}.

We can now write a CodeQL class for Java classes with the \code{@Threadsafe} annotation and give it some interesting member predicates. The first we show is \code{conflicting}, a binary predicate capturing Definition~\cref{def:conflicting_actions}:
\begin{lstlisting}[
    caption=CodeQL class ClassAnnotatedAsThreadSafe,
    label=lst:class-annotated-as-thread-safe,
    language=CodeQL]
class ClassAnnotatedAsThreadSafe extends Class {
  ClassAnnotatedAsThreadSafe() { this = annotatedAsThreadSafe() }
  /**
   * Actions `a` and `b` are conflicting iff
   * they are field access operations on the 
   * same field and at least one of them is a write.
   */
  predicate conflicting(ExposedFieldAccess a, ExposedFieldAccess b) {
    // We allow a = b, since they could be 
    // executed on different threads
    // We are looking for two operations
    // on the same field
    a.getField() = b.getField() and
    // on this class
    a.getField() = this.getAField() and
    // where at least one is a write
    // wlog we assume that is `a`
    // We use a slightly more inclusive definition 
    // than simply `a.isVarWrite()`
    Modification::isModifying(a)
  }
  ...
}
\end{lstlisting}
In line 19 we refer to the predicate \code{isModifying} in the module \code{Modification}. This predicate captures that a field access can modify a field in several ways depending on the Java type of the field. For instance, if the field is an array, a write to an index of this array would be a modification, but it would actually only contain a read of the array itself.

The conflicting field accesses are the ones where the user must ensure synchronization. We now show the member predicate \code{monitors} and its two helper predicates. These capture when a field access is only publicly accessible through synchronized means:
\begin{lstlisting}[
    caption=member predicate ClassAnnotatedAsThreadSafe::monitors,
    label=lst:member-predicate-monitors,
    language=CodeQL]
class ClassAnnotatedAsThreadSafe extends Class {
  ...
  predicate monitors(ExposedFieldAccess a, Monitors::Monitor monitor) {
    forex(Method m | this.providesAccess(m, _, a) and m.isPublic() |
      this.monitorsVia(m, a, monitor)
    )
  }

  predicate publicAccess(Expr e, ExposedFieldAccess a) {
    exists(Method m | m.isPublic() | this.providesAccess(m, e, a))
  }

  predicate providesAccess(Method m, Expr e, ExposedFieldAccess a) {
    m = this.getAMethod() and
    (
      a.getEnclosingCallable() = m and
      e = a
      or
      exists(MethodCall c | c.getEnclosingCallable() = m |
        this.providesAccess(c.getCallee(), _, a) and
        e = c
      )   )   }
  ...
}
\end{lstlisting}
We read this from below, starting with \code{providesAccess} on line 15. This captures that in method \code{m} on the class, the expression \code{e} results in the execution of the field access \code{a}. In line 16, it restricts \code{m} to be a method on the class. It then provides two alternatives: Either \code{m} contains \code{a} (line 16) and then \code{e} is \code{a} (line 19); this is the base case. Otherwise, \code{m} contains a method call (line 21) that results in the execution of \code{a} and in this case, \code{e} is that method call (line 23); this is expressed by referring to \code{providesAccess} recursively on line 22.
\looseness -1

Using \code{providesAccess} the predicate \code{publicAccess} on line 11 simply states that expression \code{e}, which is publicly accessible, results in the execution of \code{a}.

We can now read the predicate \code{monitors} on line 5. It relates a field access, \code{a}, to a monitor, \code{m}. It says that every single expression that provides public access to \code{a} is monitored by \code{m} (and also that such an expression exists, \code{forex} is a short-hand for \code{forall and exists}). This captures situations like the one below where \code{a} happens in a private
method and looks completely unsynchronized, but where the public accessors provide the synchronization:
\begin{lstlisting}
@ThreadSafe
public class Test {
  private int y;
  private Lock lock = new ReentrantLock();

  private void setYPrivate(int y) {
    this.y = y;
  }

  public void setY(int y) {
    lock.lock();
    setYPrivate(y);
    lock.unlock();
  }
}
\end{lstlisting}
We will not show the full definition of the monitor class or the predicate \code{locallyMonitors}, but we will show one building block, namely the bit that recognizes when an expression is protected by a lock (the other parts recognize protection by the various versions of the \code{synchronized} keyword and are already present in the CodeQL Java library):
\begin{lstlisting}[
    caption=locallyLockedOn,
    label=lst:locally-locked-on,
    language=CodeQL]
/** * Holds if `e` is synchronized on the `Lock` `lock` 
    * by a locking call. */
predicate locallyLockedOn(Expr e, Field lock) {
  isLockType(lock.getType()) and
  exists(Variable localLock, MethodCall lockCall, MethodCall unlockCall |
    represents(lock, localLock) and
    lockCall.getQualifier() = localLock.getAnAccess() and
    lockCall.getMethod().getName() in ["lock", "lockInterruptibly", "tryLock"] and
    unlockCall.getQualifier() = localLock.getAnAccess() and
    unlockCall.getMethod().getName() = "unlock"
  |
    dominates(lockCall.getControlFlowNode(), unlockCall.getControlFlowNode()) and
    dominates(lockCall.getControlFlowNode(), e.getControlFlowNode()) and
    postDominates(unlockCall.getControlFlowNode(), e.getControlFlowNode())  ) }
\end{lstlisting}
This predicate relates an expression \code{e} to a variable \code{lock}. It requires (on line 3) that the variable is of the Java type \code{Lock}. On line 4, it uses an existential quantifier to posit the existence of two method calls \code{lockCall} and \code{unlockCall} both of which must be calls on the variable \code{lock} (lines 5 and 7) and must be calls to appropriately called methods (lines 6 and 8). It then states on line 10 that the control flow node for \code{lockCall} must dominate the control flow node for \code{unlockCall}. This means that it is impossible to execute \code{unlockCall} without having already executed \code{lockCall}. On line 11 it is stated that the control flow node for \code{lockCall} dominates the control flow node for \code{e} meaning that \code{e} is indeed protected by the lock. Finally, on line 13, we require \code{unlockCall} to post-dominate \code{e} in the control flow graph. Post-dominance is the dual notion of dominance, meaning that it is impossible to execute \code{e} without also executing \code{unlockCall} afterwards. 

\section{Evaluation}
\label{sec:evaluation}

We evaluate the performance of our CodeQL queries by answering the following research questions:

\begin{enumerate}[label=\textbf{RQ\arabic*}]
\item Do our CodeQL queries find thread-safety errors in real-world Java code repositories? \label{rq:find_bugs}
\item Is the scalability of our CodeQL queries sufficient to analyze real-world Java code repositories? \label{rq:scalability}
\end{enumerate}

\subsubsection*{Experimental Setup}
We perform our evaluation using GitHub's Multi-Repository Variant Analysis (MRVA).
MRVA allows users to run CodeQL queries on a lists of projects hosted on GitHub.
For our evaluation, we selected the top 1000 Java projects (ordered by number of stars and dependencies).
More concretely, our evaluation comprises the analysis of 3632865 Java classes.
Out of these classes, 1992 classes are marked as \texttt{@ThreadSafe} and they are distributed over 71 source code repositories.
We consider this list a representative sample of real-world Java classes.
To the best of our knowledge, this is the largest evaluation of a concurrency analysis method for Java source code.
Due to space constraints, we refer readers to the accompanying artifact for the list of 1000 Java repositories in our evaluation.

\subsection{RQ1: Do our CodeQL queries find thread-safety errors in real-world Java code repositories?}
\label{sec:rq1}
We look into the \emph{alerts} found for each of the 1000 repositories in our evaluation.
An alert is an instance of a violation of the properties \ref{prop:no_escaping}-\ref{prop:correctly_synchronized} that our queries found.
Each alert points to the precise line in the code that violates the property.
For~\ref{prop:correctly_synchronized}, the alert points to two lines of code, one for each conflicting access.
Concretely, to answer \ref{rq:find_bugs}, we discuss: number of alerts for each repository and false positives. 

\begin{figure*}[t!]
  \centering
  \includegraphics[width=0.27\textwidth]{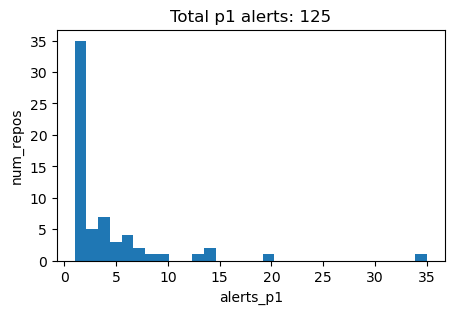}
  \includegraphics[width=0.27\textwidth]{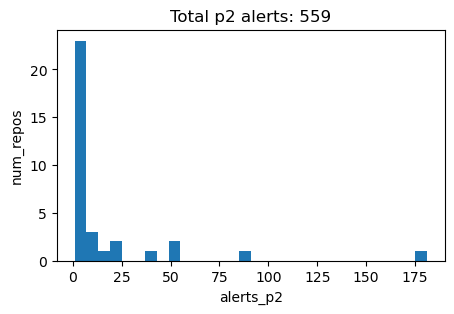}
  \includegraphics[width=0.27\textwidth]{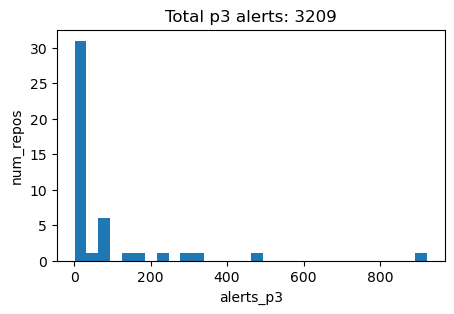}
  \caption{Histograms with distributions of number of alerts for \proponeref-\propthreeref, from left to right.}
  \label{fig:rq1}
\end{figure*}

\subsubsection*{Number of alerts}
Our queries found 3893 alerts in total.  
For~\ref{prop:no_escaping}, we found 125 alerts from 20 repositories.  
For~\ref{prop:safe_publication}, we found 559 alerts from 34 repositories.
For~\ref{prop:correctly_synchronized}, we found 3209 alerts from 45 repositories.
Property~\ref{prop:no_escaping} produces the least amount of alerts.
\Cref{fig:rq1} depicts histograms with the distribution of number of alerts.
These plots reveal that, for most repositories, we found less than \(\approx\)25 alerts.
For~\ref{prop:no_escaping} 98\% of repositories have less than 25 alerts, for~\ref{prop:safe_publication} 85\% and for~\ref{prop:correctly_synchronized} 69\%.
The reason for this unevenness in number of alerts is a few outliers in properties \proptworef and \propthreeref.
For instance, we observe a repository with more than 170 alerts for \proptworef, and several repositories with more than 200 alerts for \propthreeref.
For \propthreeref, the reason is that when a field is not correctly protected, then almost any access to the field is part of a conflicting access. 
If the field is commonly accessed in the class, then it generates a large number of alerts.
Note that although multiple alerts stem from the same concurrency issue, they are not false positives---we discuss false positives below.
Similarly, for \proptworef, when a repository has many fields in classes annotated as \texttt{@ThreadSafe} and safe publication is not systematically ensured for these fields, then we observe a larger amount of alerts.
\looseness -1

\subsubsection*{False positives}
To detect false positives, we manually inspected the alerts reported by our queries.
For~\proponeref we inspected all 125 alerts.
For \proptworef and \propthreeref, due to the large number of alerts, we inspected up to 30 alerts per repository.
This gives a high alert coverage per repository.
This is useful to find different kinds of false positives; as typically false positives are similar within the same repository.
For \proptworef, we cover all alerts for 29 (out of 34) repositories, and for the remaining repositories we have 79\%, 58\%, 58\%, 34\%, and 16\% coverage.
For \propthreeref, we cover all alerts for 32 (out of 45) repositories, and for the remaining repositories we had 45\%, 40\%, 39\%, 36\%, 35\%, 32\%, 24\%, 18\%, 13\%, 11\%, 9\%, 6\%, and 3\% coverage.

In total, we found 110 false positives (out of the 1024 alerts that we manually analyzed).
We found 1, 6 and 103 false positives for properties \proponeref-\propthreeref, respectively.
There were two reasons for all the false positives we found: fields instantiated on thread-safe classes that our queries do not know to be thread-safe, and not using standard lock APIs/classes.
The classes that our queries consider thread-safe are all classes in the package \texttt{java/util/concurrent} and selected classes from the library \texttt{com.google.common}.
Some of the false positives we found were caused by fields being instances of thread-safe classes from other libraries.
Regarding locks, our queries consider standard Java locks and the method calls from that interface \texttt{lock()}, \texttt{tryLock()}, \texttt{lockInterruptibly()}, and \texttt{unlock()} as well as \texttt{synchronized} methods and blocks.
The false positives we found regarding locking were due to the code using read-write locks or stamped locks.
These issues do not limit the applicability of our queries.
On the one hand, we designed the queries so that the list of classes to be considered thread-safe can be extended by developers. 
On the other hand, we were aware of the different locking mechanisms and APIs, and we plan to include them in future versions of the queries.
There is no fundamental limitation in CodeQL that prevents detecting the different types of locks and their APIs to lock/unlock.

To further validate the usefulness of the alerts, we reported a selection of 5 alerts as pull requests (PRs) to their corresponding repositories.
We selected alerts from 3 popular repositories Apache Flink (24.6k stars), gRPC (11.6k starts) and jib (13.9k stars).
Developers replied to all PRs, and one of the PRs has started a code restructuring discussion to fix the thread-safety issue that we reported.

\paragraph*{Conclusion for \ref{rq:find_bugs}}
We conclude that our CodeQL queries effectively find thread-safety errors in real-world Java code repositories.
We found 3893 alerts in the top 1000 Java repositories on GitHub.
Although we discovered false positives in our results, they were caused by the large diversity of thread-safe classes and locking mechanisms in the concurrent Java ecosystem.
That said, our queries can easily be extended to remove those false positives.
Finally, we reported a selection of the alerts as PRs in popular GitHub repositories.
Developers reacted positively to the PRs, and one of them triggered a code restructuring task for the project.

\subsection{RQ2: Is the scalability of our CodeQL queries sufficient to analyze real-world Java code repositories?}
\label{sec:rq2}

\begin{figure*}[t!]
  \centering
  \includegraphics[width=0.9\textwidth]{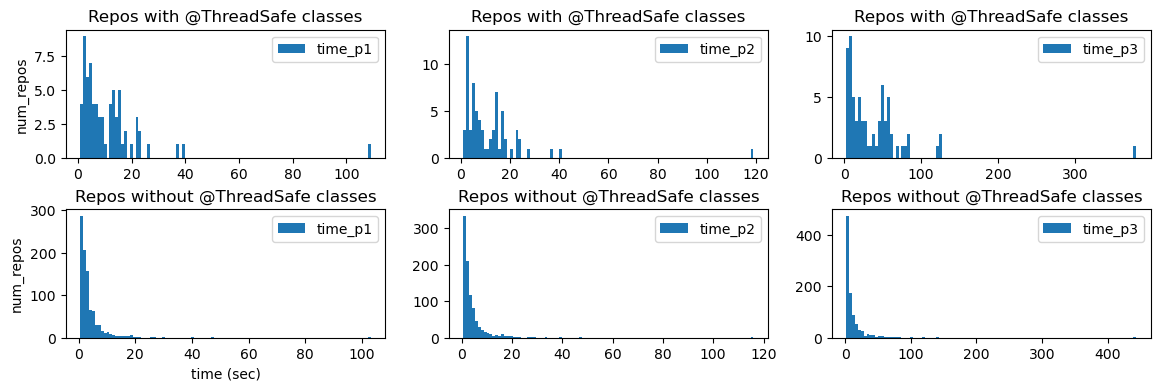}
  \includegraphics[width=1.0\textwidth,trim = 0mm 0mm 0mm 7mm, clip]{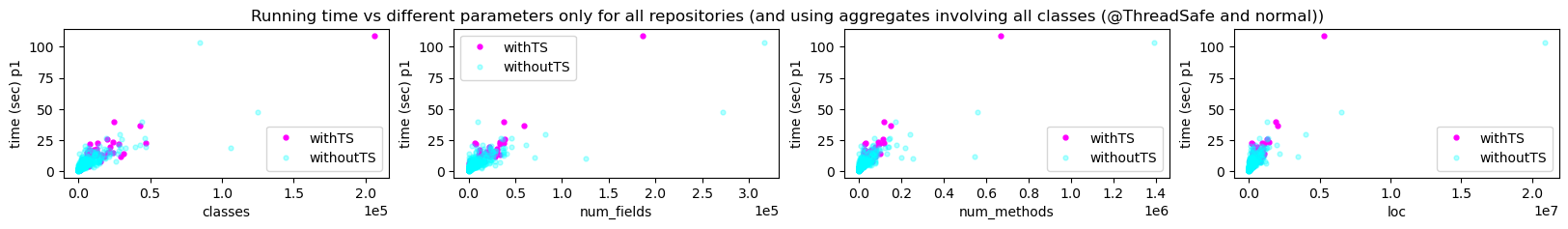}
  \includegraphics[width=1.0\textwidth,trim = 0mm 0mm 0mm 7mm, clip]{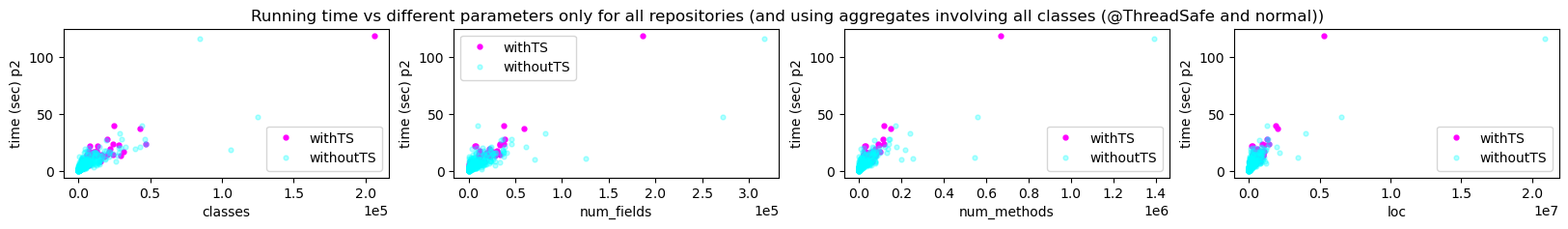}
  \includegraphics[width=1.0\textwidth,trim = 0mm 0mm 0mm 7mm, clip]{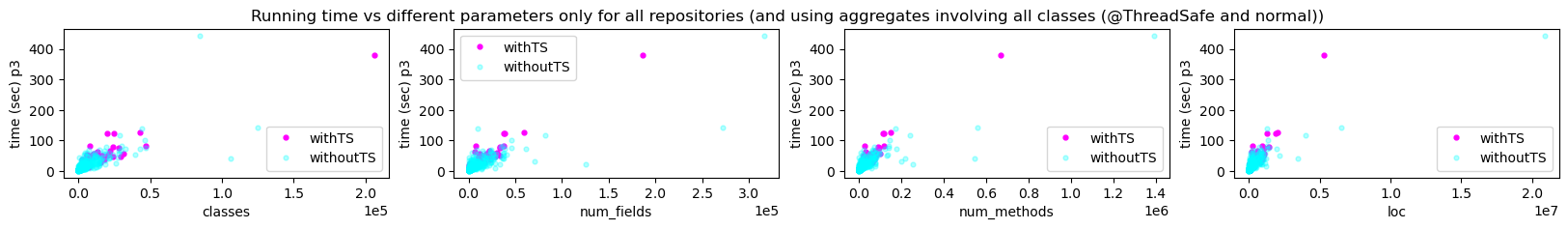}
  \caption{\underline{Rows 1-2}: Histograms with runtime distribution for \proponeref-\propthreeref (from left to right) and split depending on the repositories contain classes annotated as \texttt{@ThreadSafe}. \underline{Rows 3-5}: Repository features (number of classes, number of fields, number of methods and lines of code) versus runtime for \proponeref-\propthreeref (from top to bottom) split between repositories with classes annotated as \texttt{@ThreadSafe} (purple) and without them (light blue).
  }
  \label{fig:rq2}
\end{figure*}

To evaluate \ref{rq:scalability}, we look into the runtime for our queries on each of the 1000 projects in our evaluation.
We obtain the measurements from the MRVA analysis tool in CodeQL.
MRVA analyses are executed in virtual machines with 4 cores and 14GB of RAM running Ubuntu 22.04.
The running time that we report is representative of the running time of the queries for a project hosted on GitHub, as these specs are the same for the machines running the queries when submitting a PR in GitHub.
In what follows, we discuss the overall runtime of each query, and also look into the runtime in relation to repository features, i.e., number of classes, number of fields, number of methods and lines of code.

\subsubsection*{Overall runtime}
\Cref{fig:rq2} (first two rows) shows histograms with the distribution of runtime in seconds.
The first row depicts runtime for repositories with classes annotated as \texttt{@ThreadSafe}, and the second row repositories without classes annotated as \texttt{@ThreadSafe}.
We observe that the runtime is below 2 minutes for all repositories for \proponeref and \proptworef, and for \propthreeref it is for 99.3\% of repositories.
The slowest runtime that we measured was 444 seconds (\(\approx\)7.5 minutes) for \propthreeref.
Furthermore, for \proponeref and \proptworef the runtime is below 1 minute for 99.8\% of the repositories.
These are excellent results, as it implies that a developer could wait as little as \(\approx\)2 minutes after submitting a PR to obtain the results.
Finally, we observe a minor increase in running time in repositories with classes annotated as \texttt{@ThreadSafe}.
This is expected, as in these cases, our queries require further analysis.

\subsubsection*{Relation between runtime and repository features}
The last three rows in~\cref{fig:rq2} plot runtime versus the following repository features: number of classes, number of fields, number of methods and lines of code (loc).
We separate data points from repositories containing classes annotated as \texttt{@ThreadSafe} (in purple), and repositories without them (in light blue).
The goal of this experiment is to determine whether there are correlations between these features and the runtime.
As before, we observe a minor distinction between repositories containing classes annotated as \texttt{@ThreadSafe} classes and those who do not; with the former having slightly larger running times.
The plots show a positive linear correlation between runtime and all features, for all properties.
We observe that number of methods and lines of code have a stronger increasing effect on runtime than number of classes or number of fields.
These plots also shed light on the outliers from the previous discussion.
For instance, the repository with the 7.5 minutes runtime contains 20 million lines of code, 1.4 million methods, 300k fields and 100k classes.
We can also observe that for repositories with up to 200k lines of code, 20k methods, 6000 fields, and 1200 classes, the runtime is below 2 minutes.
These correspond to 99.3\% of the repositories in our evaluation.

\paragraph*{Conclusion for \ref{rq:scalability}}
We conclude that the scalability of our CodeQL queries is sufficient to analyze real-world Java code repositories.
Our evaluation shows that the queries run below 2 minutes in repositories with thousands of classes, fields, methods and lines of code.
Since our queries are executed when submitting a PR, we argue that a runtime of 2 minutes does not significantly disturb the workflow of developers.
The evaluation also indicates that runtime grows linearly with respect to the features mentioned above; with lines of code and number of methods being the features that have the strongest effect.
\looseness -1


\section{Related Work}

\subsection{Static Analysis}
Habib and Pradel present TSFinder, a method to automatically determine whether a Java class is thread-safe~\cite{habib_is_2018}.
TSFinder uses a machine learning classifier to perform the analy-sis.
Java classes are transformed into graphs, which are used as input to the classifier.
The analysis by TSFinder returns a label denoting whether the class under analysis is thread-safe.
Instead, our method relies on precise definitions for thread-safety founded on the Java memory model.
As a consequence, when we report an alert, our method points to the program statement(s) causing the violation (instead of giving a binary answer).
AutoLock is a static analysis method to automatically add locks in Java classes to fix data races~\cite{wang_towards_2020}.
This method uses heuristics based on data-flow and alias analysis to detect unprotected mutable fields.
Our queries capture other important aspects of thread-safety for classes, namely, preventing field escaping and safe publication.
This is due to re-using the correctness notions in the Java memory model for our analysis.
\looseness -1

There exist a line of work on deductive verification that focuses on correctness of concurrent software~\cite{DBLP:conf/esop/LeinoM09,DBLP:conf/fosad/LeinoMS09,DBLP:conf/plpv/AmighiBHZ12}.
These works translate programs into logical formulae and use Satisfiability Modulo Theory (SMT) solvers to perform the analysis.
Typically using permission based logics.
Permission logics have also been used for type-based analysis~\cite{DBLP:conf/oopsla/BierhoffA07}.
Deductive verification methods require class and method specifications in the form of pre- and post-conditions to perform the analysis.
This allows for checking a larger range of concurrency issues compared to our queries.
However, these methods require specialized knowledge.
They require developers to learn the logics used for reasoning and specify pre- and post-conditions for all classes and methods under analysis.
On the contrary, our method only requires the annotation \texttt{@ThreadSafe} on the target class; this is a standard and widely used annotation in Java (as~\cref{sec:evaluation} shows).
Moreover, the effort required to add pre- and post-conditions to thousands of classes and methods---as several of the projects in our evaluation have---could be a deterrent for the use of these methods in practice.

The clang C/C++ compiler has an integrated thread-safety analysis tool~\cite{hutchins_cc_2014}.
Similar to the above, this tool requires multiple annotations to find concurrency errors.
It requires that developers add annotations regarding what lock is being used in different parts of the code.
Note that this is a tedious and error-prone task.
Our method automatically determines the relevant locks for each operation using the happens-before order.
Although our queries do not capture all lock classes and lock/unlock APIs, this is not a prohibiting limitation; CodeQL does not have foundational restrictions that would prevent the adding of these classes and lock/unlock APIs.
\looseness -1

\subsection{Testing and Model-Checking}
Testing methods have been extensively applied to find concurrency errors~\cite{nistor2012ballerina,choudhary2017efficient,DBLP:conf/pldi/PradelG12,DBLP:conf/asplos/BurckhardtKMN10,DBLP:conf/asplos/ZhaoW0R25}.
A line of work focuses on generation of concurrent tests~\cite{choudhary2017efficient,nistor2012ballerina,DBLP:conf/pldi/PradelG12}.
These methods automatically generate tests involving multiple threads that are designed to target thread-safety errors.
Although effective, these methods do not ensure that thread-safety violations will be uncovered, as this depends on the runtime scheduler executing the right interleaving of operations.
To diminish this problem, another line of work focuses on designing testing methods with probabilistic guarantees.
Probabilistic concurrency testing~\cite{DBLP:conf/asplos/BurckhardtKMN10,DBLP:conf/asplos/ZhaoW0R25} aims at ensuring that all interleavings of operations are executed with non-zero probability.
In this way, it is possible to probabilistically guarantee that if a class have thread-safety errors, the interleaving triggering an error will be executed with non-zero probability.
Finally, other methods such as Coyote~\cite{deligiannis_industrial-strength_2023} attempt to systematically explore all possible interleavings of a test to find thread-safety errors.
As opposed to testing methods, the accuracy of our method does not depend on the number of interleavings or their likelihood to be executed by the scheduler.

Model-checking is a verification method designed to systematically explore all possible executions of a program.
Model-checking has been applied to detect thread-safety errors~\cite{DBLP:conf/cav/WuHHLS23, schemmel2020symbolic, aronis2018optimal}.
To avoid the state explosion problem, these methods require state-reduction techniques to decrease the state space size.
Similarly to testing, the main difference of these methods and our queries is that the performance of our analysis does not depend on the number of executions.
This is crucial to be able to efficiently analyze code repositories with thousands of classes, fields, methods and lines of code (which we demonstrated that our method can handle, cf.~\cref{sec:evaluation}).
\looseness -1


\section{Conclusion}
We have presented a scalable method to automatically detect thread-safety violations in real-world Java classes.
We have introduced a definition of thread-safe class for Java that is founded on the notion of correctness of the Java memory model, namely, data race freedom.
We have introduced 3 properties that determine whether a Java classes is thread-safe.
We have proven that these properties are sufficient to achieve thread-safety.
We have implemented 3 queries in CodeQL, one for each of the properties.
We have evaluated our queries on the top 1000 Java repositories on GitHub; targeting classes annotated as \texttt{@ThreadSafe}.
Our evaluation comprises 3632865 Java classes; which include 1992 classes marked as \texttt{@ThreadSafe}.
Our evaluation found thousands of thread-safety errors, demonstrating that our queries find bugs in real-world Java repositories.
The runtime of our queries was below 2 minutes for repositories including up to 200k lines of code, 20k methods, 6000 fields, and 1200 classes.
The slowest runtime was 7.5 minutes for a repository with 20 million lines of code, 1.4 million methods, 300k fields and 100k classes.
These results demonstrate that the queries can be used for very large Java repositories.
We reported a selection of discovered bugs as PRs, and developers reacted positively to them.
Finally, we have submitted the queries to CodeQL, and they are currently in the process of becoming part of the standard CodeQL queries.
All in all, these results show a practical application of the Java memory model to build a highly scalable thread-safe analysis method for Java classes.


\bibliographystyle{IEEEtran}
\bibliography{references}

\end{document}